\documentclass{llncs}
\usepackage{llncsdoc}
\usepackage{amsmath}
\usepackage[ruled,vlined]{algorithm2e}

\usepackage{graphicx}
\usepackage{subfigure}

\begin{document}

\newcounter{save}\setcounter{save}{\value{section}}
{\def\addtocontents#1#2{}%
\def\addcontentsline#1#2#3{}%
\def\markboth#1#2{}%
\title{Energy-Efficient Scheduling with Time and Processors Eligibility Restrictions}

\author{Xibo Jin, Fa Zhang, Ying Song, Liya Fan and Zhiyong Liu}

\institute{Institute of Computing Technology, University of Chinese Academy of Sciences, Beijing, China\\
\email{\{jinxibo, zhangfa, songying, fanliya, zyliu\}@ict.ac.cn}
}

\maketitle
\begin{abstract}
While previous work on energy-efficient algorithms focused on assumption that tasks can be assigned to any processor, we initially study the problem of task scheduling on restricted parallel processors. The objective is to minimize the overall energy consumption while speed scaling (SS) method is used to reduce energy consumption under the execution time constraint (Makespan $C_{max}$). In this work, we discuss the speed setting in the continuous model that processors can run at arbitrary speed in $[s_{min},s_{max}]$. The energy-efficient scheduling problem, involving task assignment and speed scaling, is inherently complicated as it is proved to be NP-Complete. We formulate the problem as an Integer Programming (IP) problem. Specifically, we devise a polynomial time optimal scheduling algorithm for the case tasks have a uniform size. Our algorithm runs in $O(mn^3logn)$ time, where $m$ is the number of processors and $n$ is the number of tasks. We then present a polynomial time algorithm that achieves an approximation factor of $2^{\alpha-1}(2-\frac{1}{m^{\alpha}})$ ($\alpha$ is the power parameter) when the tasks have arbitrary size work. Experimental results demonstrate that our algorithm could provide an efficient scheduling for the problem of task scheduling on restricted parallel processors.
\end{abstract}
\section{Introduction}
Energy consumption has become an important issue in the parallel processor computational systems. Dynamic Speed Scaling (SS) is a popular approach for energy-efficient scheduling to significantly reduce energy consumption by dynamically changing the speeds of the processors. The well-known relationship between speed and power is the cube-root rule, more precisely, that is the power of a processor is proportional to $s^3$ when it runs at speed $s$ [1, 2]. Most research literatures [3, 4, 5, 6, 7, 8, 9, 10] have assumed a more general power function $s^{\alpha}$, where $\alpha>1$ is a constant power parameter. Note that it is a convex function of the processor's speed. Obviously, energy consumption is the power integrated over duration time. Higher speeds allow for faster execution, at the same time, result in higher energy consumption.

In the past few years, energy-efficient scheduling has received much attention from single processor to parallel processors environment. In the algorithmic community, the approaches can (in general) be categorized into the following two classes for reducing energy usage [5, 7]. (1) \textit{Dynamic speed scaling}: The processors lower down the speed to execute tasks as much as possible while fulfil their timing constraints. The reason behind energy saving via this strategy is the convexity of the power function. The gold is to decide the processing speeds in a way that minimizes the total energy consumption and guarantees the prescribed deadline. (2) \textit{Power-down management}: The processors will be put into the power-saving state when they are idle. But it is energy-cost for transiting back to the active state. This strategy is to determine whether there exist idle periods that can outweigh the transition cost and decide when to wake the power-saving mode in order to complete all tasks in time. Our paper focuses on energy-efficient scheduling via dynamic speed scaling strategy. In this policy, the goals of scheduling are either to minimize the total energy consumption or to trade off the conflicting objectives of energy and performance. The main difference is that the former one reduces the total energy consumption as long as the timing constraint is not violated, while the later one seeks the best point between the energy cost and performance metric (such as makespan and flow time).

Speed scaling has been widely studied to save energy consumption initiated by Yao et al. [3]. The previous work consider that a task can be assigned to any processor. But it is natural to consider the restricted scheduling in modern computational systems. The reason is that the systems evolve over time, such as cluster, then the processors of the system differ from each other in their functionality (For instance, the processors have different additional components). This leads to the task can only be assigned to the processors, which has the task's required component. I.e., it leads to different affinities between tasks and processors. In practice, certain tasks may have to be allocated for certain physical resources (such as GPU) [11]. It is also pointed out that some processors whose design is specialized for particular types of tasks, then tasks should be assigned to a processor best suited for them [12]. Furthermore, when considering tasks and input data, tasks need to be assigned on the processors containing their input data. In other words, a part of tasks can be assigned on processors set $A_i$, and a part of tasks can be assigned on processors set $A_j$, but $A_i{\neq}A_j, A_i{\cap}A_j{\neq}{\emptyset}$. Another case in point is the scheduling with processing processor restrictions aimed at minimizing the makespan has been studied extensively in the algorithmic community (See [13] for an excellent survey). Therefore, it is significant to study the scheduling with processor restrictions from both of practical and algorithmic requirements.

\textbf{Previous Work: }Yao et al. [3] were the first to explore the problem of scheduling a set of tasks with the smallest amount of energy on single processor environment via speed scaling. They proposed an optimal offline greedy algorithm and two bounded online algorithms named \textit{Optimal Available} and \textit{Average Rate}. Ishihara et al. [4] formulated the minimization-energy of dynamical voltage scheduling (DVS) as an integer linear programming problem when all tasks were ready at the beginning and shared common finishing time. They showed that in the optimal solution a processor only runs at two adjacent discrete speeds when it can use only a small number of discrete processor speeds.

Besides studying variant of the speed scaling problems on single processor, researchers also carried out studies on parallel processors environment. Chen et al. [6] considered energy-efficient scheduling with and without task migration over multiprocessor. They proposed approximation algorithm for different settings of power characteristics where no task was allowed to migrate. When task migration is allowed and migration cost is assumed being negligible, they showed that there is an optimal real-time task scheduling algorithm. Albers et al. [7] investigated the basic problem of scheduling a set of tasks on multi-processor settings with an aim to minimize the total energy consumption. First they studied the case that all tasks were unit size and proposed a polynomial time algorithm for agreeable deadlines. They proved it is NP-Hard for arbitrary release time and deadlines and gave a $\alpha^{\alpha}2^{4\alpha}$-approximation algorithm. For scheduling tasks with arbitrary processing size, they developed constant factor approximation algorithms. Aupy et al. [2] studied the minimization of energy on a set of processors for which the tasks assignment had been given. They investigated different speed scaling models. Angel et al. [10] consider the multiprocessor migratory and preemptive scheduling problem with the objective of minimizing the energy consumption. They proposed an optimal algorithm in the case where the jobs have release dates, deadlines and the power parameter ${\alpha}>2$.

There were also some literatures to research the performance under an energy bounded. Pruhs et al. [8] discussed the problem of speed scaling to optimize makespan under an energy budget in a multiprocessor environment where the tasks had precedence constraints ($Pm |prec,energy| C_{max}$, $m$ is the number of processors). They reduced the problem to the $Qm |prec| C_{max}$ and obtained a poly-log$(m)$-approximation algorithm assuming processors can change speed continuously over time. The research by Greiner et al. [9] was a present to study the trade off between energy and delay, i.e., their objective was to minimize the sum of energy cost and delay cost. They suggested a randomized algorithm $\mathcal{RA}$ for multiple processors: each task was assigned uniformly at random to the processors, and then the single processor algorithm $\mathcal{A}$ was applied separately to each processor. They proved that the approximation factor of $\mathcal{RA}$ was ${\beta}B_{\alpha}$ without task migration when $\mathcal{A}$ was a $\beta$-approximation algorithm ($B_{\alpha}$ is the $\alpha$-th Bell number). They also showed that any $\beta$-competitive online algorithm for a single processor yields a randomized ${\beta}B_{\alpha}$-competitive online algorithm for multiple processors without migration. Using the method of conditional expectations, the results could be transformed to a derandomized version with additional running time. Angel et al. [10] also extended their algorithm, which considered minimizing the energy consumption, to obtain an optimal algorithm for the problem of maximum lateness minimization under a budget of energy.

However, all of these results were established without taking into account the restricted parallel processors. More formally, let the set of tasks $\mathcal{J}$ and the set of processors $\mathcal{P}$ construct a bipartite graph $G=(\mathcal{J}+\mathcal{P},E)$, where the edge of $E$ denotes a task can be assigned to a processor. The previous work study $G$ is a \textit{complete} bipartite graph, i.e., for any two vertices, $v_1{\in}\mathcal{J}$ and $v_2{\in}\mathcal{P}$, the edge $v_1v_2$ is in $G$. We study the energy-efficient scheduling that $G$ is a \textit{general} bipartite graph, i.e., $v_1v_2$ may be not an edge of $G$.

\textbf{Our contribution: }In this paper, we address the problem of task \textbf{S}cheduling with the objective of \textbf{E}nergy \textbf{M}inimization on \textbf{R}estricted \textbf{P}arallel \textbf{P}rocessors (SEMRPP). It assumes all tasks are ready at time 0 and share a common deadline (a real-time constraint) [2, 4, 6, 7]. In this work, We discuss the continuous speed settings that processors can run at arbitrary speed in $[s_{min},s_{max}]$. We propose an optimal scheduling algorithm when all the tasks have uniform computational work. For the general case that the tasks have non-uniform computational work we prove that the minimization of energy is NP-Complete in the strong sense. We give a $2^{\alpha-1}(2-\frac{1}{m^{\alpha}})$-approximation algorithm, where $\alpha$ is the power parameter and $m$ is the number of processors. The performance of the approximation algorithm is evaluated through a set of experiments after algorithm analysis, and it turns out effective results to confirm the proposed scheduling work efficiently. To the best of our knowledge, our work may be the initial attempt to study energy optimization on the restricted parallel processors.

The remainder of this paper is organized as follows. We provide the formal description of model in Sections 2. Section 3 discusses some preliminary results and formulate the problem as an Integer Programming (IP) problem. In Section 4, we devise a polynomial time optimal scheduling algorithm in the case where the tasks have uniform size. In Section 5, we present a bounded factor approximation guarantee algorithm for the general case that the tasks have arbitrary size work. Section 6 shows the experimental results. Finally we conclude the paper in Sections 7.
\section{Problem and Model}
We model the SEMRPP problem of scheduling a set $\mathcal{J}=\{J_1,J_2,...,J_n\}$ of $n$ independent tasks on a set $\mathcal{P}=\{P_1,P_2,...,P_m\}$ of $m$ processors. Each task $J_j$ has an amount of computational work $w_j$ which is defined as the number of the required CPU cycles for the execution of $J_j$ [3]. We refer to the set $\mathcal{M}_j\subseteq{\mathcal{P}}$ as eligibility processing set of the task $J_j$, that is, $J_j$ needs to be scheduled on one of its eligible processors $\mathcal{M}_j(\mathcal{M}_j\neq{\phi})$. We also say that $J_j$ is allowable on processor $P_i\in{\mathcal{M}_j}$, and is not allowed to migrate after it is assigned on a processor. A processor can process at most one task at a time and all processors are available at time $0$.

At any time $t$, the speed of $J_j$ is denoted as $s_{jt}$, and the corresponding processing power is $P_{jt}=(s_{jt})^\alpha$. The amount of CPU cycles $w_j$ executed in a time interval is the speed integrated over duration time and energy consumption $E_j$ is the power integrated over duration time, that is, $w_j=\int{s_{jt}dt}$ and $E_j=\int{P_{jt}dt}$, following the classical models of the literature [2, 3, 4, 5, 6, 7, 8, 9, 10]. Note that in this work we focus on speed scaling and all processors are alive during the whole execution, so we do not take static energy into account [2, 8]. Let $c_j$ be the time when the task $J_j$ finishes its execution. Let $x_{ij}$ be an $0-1$ variable which is equal to one if the task $J_j$ is processed on processor $P_i$ and zero otherwise. We note that $x_{ij}=0$ if $P_i\notin{\mathcal{M}_j}$. Our goal is scheduling the tasks on processors to minimize the overall energy consumption when each task could finish before the given common deadline $C$ and be processed on its eligible processors. Then the SEMRPP problem is formulated as follows:
\[\mathbf{(P_0)}\qquad min{\sum_{j=1}^n\int{P_{jt}dt}}\]
\[s.t. \qquad c_j\leq{C}  \quad \forall{J_j},\]
\[\sum_{i=1}^mx_{ij}=1  \quad \forall{J_j},\]
\[x_{ij}{\in}\{0,1\}  \quad \forall{J_j},P_i\in{\mathcal{M}_j},\]
\[x_{ij}=0  \quad \forall{J_j},P_i\notin{\mathcal{M}_j}.\]
\section{Preliminary Lemma}
We start by giving preliminary lemmas for reformulating the SEMRPP problem.
\begin{lemma}
If $S$ is an optimal schedule for the SEMRPP problem in the continuous model, it is optimal to execute each task at a unique speed throughout its execution.
\end{lemma}
\begin{proof}
Suppose $S$ is an optimal schedule that some task $J_j$ does not run at a unique speed during its execution. We denote $J_j$'s speeds by $s_{j1},s_{j2},...,s_{jk}$, the power of each speed $i$ is $(s_{ji})^\alpha,i=(1,2,...,k)$, and the execution time of the speeds are $t_{j1},t_{j2},...,t_{jk}$, respectively. So, its energy consumption is $\sum_{i=1}^kt_{ji}(s_{ji})^\alpha$. We average the $k$ speeds and keep the total execution time unchanged, i.e., $\bar{s}_j=(\sum_{i=1}^ks_{ji}t_{ji})/(\sum_{i=1}^kt_{ji})$. Because the power function is a convex function of speed, according to convexity [14] (In the rest of paper, it will use convexity in many place but will not add reference [14]), we have
\begin{equation*}
\begin{split}
\sum_{i=1}^kt_{ji}(s_{ji})^\alpha
&{=}(\sum_{i=1}^kt_{ji})(\sum_{i=1}^k\frac{t_{ji}}{\sum_{i=1}^kt_{ji}}(s_{ji})^{\alpha})\\
&{\geq}(\sum_{i=1}^kt_{ji})(\sum_{i=1}^k\frac{t_{ji}s_{ji}}{\sum_{i=1}^kt_{ji}})^{\alpha}
=(\sum_{i=1}^kt_{ji})(\bar{s}_j)^{\alpha}\\
&{=}\sum_{i=1}^kt_{ji}{(\bar{s}_j)^{\alpha}}
\end{split}
\end{equation*}
So the energy consumption by unique speed is less than a task run at different speeds. I.e. , if we do not change $J_j$'s execution time and its assignment processor (satisfying restriction), we can get a less energy consumption scheduling, which is a contradiction to that $S$ is an optimal schedule.
\end{proof}
\begin{corollary}
There exists an optimal solution for SEMRPP in the continuous model, for which each processor executes all tasks at a uniform speed, and finishes its tasks at time $C$.
\end{corollary}
All tasks on a processor run at a unique speed can be proved like \emph{Lemma 1}. If some processor finishes its tasks earlier than $C$, it can lower its speed to consume less energy without breaking the time constraint and the restriction. Furthermore there will be no gaps in the schedule [8].

Above discussion leads to a reformulation of the SEMRPP problem in the continuous model as following:
\[\mathbf{(P_1)}\qquad min{\sum_{i=1}^m\frac{(\sum\limits_{j=1}^nx_{ij}w_j)^\alpha}{C^{\alpha-1}}}\]
\begin{eqnarray}
s.t. \qquad \sum_{j=1}^nx_{ij}w_j\leq{s_{max}C}  \quad \forall{P_i},\\
\sum_{i=1}^mx_{ij}=1  \quad \forall{J_j},\\
x_{ij}{\in}\{0,1\}  \quad \forall{J_j},P_i\in{\mathcal{M}_j},\\
x_{ij}=0  \quad \forall{J_j},P_i\notin{\mathcal{M}_j}.
\end{eqnarray}
The objective function is from that a processor $P_i$ runs at speed $\frac{\Sigma_{J_j on P_i}w_j }{C}=\frac{\Sigma_{j=1}^nx_{ij}w_j}{C}$, that is each task on $P_i$ will run at this speed, and $P_i$ will complete all the tasks on it at time $C$ (It assumes that, in each problem instance, the computational cycles of the tasks on one processor is enough to hold the processor will not run at speed $s_i<s_{min}$. Otherwise we are like to turn off some processors). Constraint $(1)$ follows since a processor can not run at a speed higher than $s_{max}$. Constraint (2) relates to that if a task has assigned on a processor it will not be assigned on other processors, i.e, non-migratory. Constraint (3) and (4) are the restrictions of the task on processors.
\begin{lemma}
Finding an optimal schedule for SEMRPP problem in the continuous model is NP-Complete in the strong sense.
\end{lemma}
\begin{proof}
First, we transform the optimization problem to an associated decision problem: given time and processors eligibility restrictions, and a bound on the energy consumption, is there a schedule such that the restrictions and the bound on energy consumption are satisfied. Clearly, it is in NP, since we can verify in polynomial time that a proposed schedule satisfies the given restrictions and the bound on energy consumption. We will prove that finding an optimal schedule for SEMRPP problem is NP-Complete in the strong sense via the reduction to the 3-PARTITION problem.

Consider an instance of the SEMRPP problem that $\mathcal{M}_j=\mathcal{P}$ for all tasks $J_j$ and $s_{max}$ is fast enough to assure a feasible schedule for the given tasks. By the convexity of the function $f(s)=s^\alpha(\alpha>1)$, we note that the optimal schedule is to averagely partition the tasks to processors. Then we can finish the proof by a pseudo-polynomial reduction from the 3-PARTITION problem.

Consider an instance of 3-Partition: Given a list $A=(a_1,a_2,...,a_{3m})$ of $3m$ positive integers such that ${\sum}a_j=mB, \frac{1}{4}<a_j<\frac{1}{2}$ for each $1{\leq}j{\leq}3m$, is there a partition of $A$ into $A_1,A_2,...,A_m$ such that $\sum_{a_j{\in}A_i}a_j=B$ for each $1{\leq}i{\leq}m$? [15, 16] We construct an instance of SEMRPP problem as follows. There are $3m$ tasks for whose execution cycles are equal to $a_j$ and there are $m$ processors. The deadline $C=1$ and the energy consumption is $mB^{\alpha}$. Denote the execution cycles of processors as $(h_1,h_2,...,h_m)$. According to $\mathbf{(P_1)}$, the energy consumption is ${\sum}_{i=1}^m(h_i)^{\alpha}$. By convexity, we have ${\sum}_{i=1}^m(h_i)^{\alpha}=m{\sum}_{i=1}^m\frac{1}{m}(h_i)^{\alpha}{\geq}m(\frac{1}{m}{\sum}_{i=1}^mh_i)^{\alpha}=mB^{\alpha}$ (Note that ${\sum}_{i=1}^mh_i=mB$). The energy consumption is equal to $mB^{\alpha}$ if and only if $h_1=h_2=...=h_m=B$. Thus, there is an optimal schedule if and only if there is a 3-Partition. It is clear that the above reduction is a pseudo-polynomial reduction. So we can conclude that SEMRPP in the continuous model is strongly NP-Complete by this pseudo-polynomial time reduction to 3-PARTITION problem which has been proved NP-Complete in the strong sense.
\end{proof}
\begin{lemma}
There exists a polynomial time approximation scheme (PTAS) for the SEMRPP problem in the continuous model, when $\mathcal{M}_j=\mathcal{P}$ and $s_{max}$ is fast enough.
\end{lemma}
\begin{proof}
The proof is a little similar to [8] whose aim is giving a PTAS for the problem that measures the makespan under an energy bounded $(Sm |energy| C_{max})$. It turns out that the SEMRPP problem is equivalent to minimizing the $\textit{l}_{\alpha}$ norm \footnote{For a positive number $\alpha{\geq}1$, the $\textit{l}_{\alpha}$ norm of a vector $\textbf{x}=(x_1,x_2,...,x_n)$ is defined by $\|\textbf{x}\|=(|x_1|^{\alpha}+|x_2|^{\alpha}+...+|x_n|^{\alpha})^{\frac{1}{\alpha}}$} of the loads [17] from the description of \textit{Lemma 2} (see ${\sum}_{i=1}^m(h_i)^{\alpha}$ and $\alpha$ is a constant power parameter). Then we use the PTAS given in [17], that is, for any $\epsilon>0$, we can find the sum of the execution cycles of the tasks on processor $P_i$ (denoted as load below) $L_1,L_2,...,L_m$ in polynomial time such that $\Sigma_{i=1}^m(L_i)^\alpha{\leq}(1+\epsilon)\Sigma_{i=1}^m(OPT_i)^\alpha$, where $L_i$ is the load of scheduling and $OPT_i$ is the optimal load for processor $P_i$, respectively.
\end{proof}
Note that we give the detail proof of Lemma 2 and Lemma 3 that were similarly stated as observations in the work [7], and we mainly state the conditions when they are established in the restricted environment. (such as the set of restricted processors and the upper speed $s_{max}$ that we discuss below in the paper)
\section{Uniform tasks}
We now propose an optimal algorithm for a special case of SEMRPP problem for which all tasks have equal execution cycles (uniform) (denoted as ECSEMRPP\_Algo algorithm). Note that we can set $w_j=1,{\forall}J_j$ and set $C=C/w_j$ in $\mathbf{(P_1)}$ without loss of generality. Given the set of tasks $\mathcal{J}$, the set of processors $\mathcal{P}$ and the sets of eligible processors of tasks $\{\mathcal{M}_j\}$, we construct a network $G=(V,E)$ as follow: the vertex set of $G$ is $V=\mathcal{J}\cup\mathcal{P}\cup\{s,t\}$ ($s$ and $t$ correspond to a source and a destination, respectively), the edge set $E$ of $G$ consists of three subsets: (1)$(s,P_i)$ for all $P_i{\in}\mathcal{P}$; (2)$(P_i,J_j)$ for $P_i{\in}\mathcal{M}_j$; (3)$(J_j,t)$ for all $J_j{\in}\mathcal{J}$. We set unit capacity to edges $(P_i,J_j)$ and $(J_j,t)$, $(s,P_i)$ have capacity $c$ (initially we can set $c=n$). Define $L^*=min\{max\{L_i\}\} (i=1,2,...,m)$, $L_i$ is the load of processor $P_i$ and it can be achieved by \textbf{\emph{Algorithm 1}}.
\begin{algorithm}
\SetKwInOut{Input}{input}\SetKwInOut{Output}{output}
\Input{$(G,n)$}
\Output{$L^*,P_i$ that have the maximal load, the set $\mathcal{J}_i$ of tasks that load on $P_i$}
1: Let variable $l=1$ and variable $u=n$\;
2: If $l=u$, then the optimal value is reached: $L^*=l$, return the $P_i$ and $\mathcal{J}_i$, stop\;
3: Else let capacity $c={\lfloor}\frac{1}{2}(l+u){\rfloor}$. Find the Maximum-flow in the network $G$. If the value of Maximum-flow is exact $n$, namely $L^*{\leq}c$, then set $u=c$ and keep $P_i$, $\mathcal{J}_i$ by the means of the Maximum-flow. Otherwise, the value of Maximum-flow is less than $n$, namely $L^*>c$, we set $l=c+1$. Go back to $2$.
\caption{{BS\_Algo$(G,n)$}}
\end{algorithm}
\begin{lemma}
The algorithm BS\_Algo solves the problem of finding minimization of maximal load of processor for restricted parallel processors in $O(n^3logn)$ time, if all tasks have equal execution cycles.
\end{lemma}
Its proof can mainly follow from the Maximum-flow in [18]. The computational complexity is equal to the time $O(n^3)$ to find Maximum-flow multiple $logn$ steps, i.e, $O(n^3logn)$.

We construct our ECSEMRPP\_Algo algorithm (\textbf{\emph{Algorithm 2}}) through finding out the min-max load vector $\vec{l}$ that is a strongly-optimal assignment defined in [17, 19].
\begin{definition}
Given an assignment $H$ denote by $S_k$ the total load on the $k$ most load of processors. We say that an assignment is {strongly-optimal} if for any other assignment $H^{'}$ ($S_k^{'}$ accordingly responds to the total load on the $k$ most load of processors) and for all $1{\leq}k{\leq}m$ we have $S_k{\leq}S_k^{'}$.
\end{definition}
\begin{algorithm}
1: Let $G_0=G(V,E)$, $n_0=n$, $\mathcal{P}^H=\phi$, $\mathcal{J}^H=\{\phi_1,...,\phi_m\}$\;
2: Call $BS\_Algo(G_0,n_0)$\;
3: Set maximal load sequence index $i=i+1$. According to the scheduling returned by step $2$, we note the processor $P_i^H$ that have actual maximal load and note its task set $\mathcal{J}_i^H$. $\mathcal{E}_i^H$ corresponds to the related edges of $P_i^H$ and $\mathcal{J}_i^H$. We set $G_0=\{V{\setminus}{P_i^H}{\setminus}{\mathcal{J}_i^H},E{\setminus}{\mathcal{E}_i^H}\}$, $\mathcal{P}^H=\mathcal{P}^H{\cup}\{P_i^{H}\}$, ${\phi}_i=\mathcal{J}_i^{H}$. We set $n_0=n_0-|{\mathcal{J}_i^H}|$. If $G_0{\neq}{\phi}$, go to step 2\;
4: We assign the tasks of $\mathcal{J}_i^H$ to $P_i^H$ and set all tasks at speed $\frac{\Sigma_{J_j{\in}\mathcal{J}_i^H}w_j}{C}$ on $P_i^H$. Return the final schedule $H$.
\caption{ECSEMRPP\_Algo}
\end{algorithm}
\begin{theorem}
Algorithm ECSEMRPP\_Algo finds the optimal schedule for the SEMRPP problem in the continuous model in $O(mn^3logn)$ time, if all tasks have equal execution cycles.
\end{theorem}
\begin{proof}
First we prove the return assignment $H$ of ECSEMRPP\_Algo is a strongly-optimal assignment. We set $H=\{L_1,L_2,...,L_m\}$, $L_i$ corresponds to the load of processor $P_i$ in non-ascending order. Suppose $H^{'}$ is another assignment that $H^{'}{\neq}H$ and $\{L_1^{'},L_2^{'},...,L_m^{'}\}$ corresponds to the load. According to the ECSEMRPP\_Algo algorithm, we know that $H^{'}$ can only be the assignment that $P_i$ moves some tasks to $P_j(j<i)$, because $P_i$ can not move some tasks to $P_{j^{'}}(j^{'}{>}i)$ otherwise it can lower the $L_i$ which is a contradiction to ECSEMRPP\_Algo algorithm. We get $\Sigma_{k=1}^{i}L_i{\leq}\Sigma_{k=1}^{i}L_i^{'}$, i.e., $H$ is a strongly-optimal assignment by the definition. It turns out that there does not exist any assignment that can reduce the difference between the loads of the processors in the assignment $H$. I.e., there are not other assignment can reduce our aim as it is convexity. So the optimal scheduling is obtained.

Every time we discard a processor, so the total cost time is $m{\times}O(n^3logn)=O(mn^3logn)$ according to \textit{Lemma $4$}, which completes the proof.
\end{proof}
\section{General tasks}
As it is NP-Complete in the strong sense for general tasks (\textit{Lemma $2$}), we aim at getting an approximation algorithm for the SEMRPP problem. First we relax the equality $(3)$ of $\mathbf{(P_1)}$ to
\begin{equation}
0{\leq}x_{ij}{\leq}1 \qquad {\forall}J_j,P_i{\in}\mathcal{M}_j
\end{equation}

After relaxation, the SEMRPP problem transforms to a convex program. The feasibility of the convex program can be checked in polynomial time to within an additive error of ${\epsilon}$ (for an arbitrary constant ${\epsilon}>0$) [20], and it can be solved optimally [14]. Suppose $x^*$ be an optimal solution to the relaxed SEMRPP problem. Now our goal is to convert this fractional assignment to an integral one $\bar{x}$. We adopt the dependent rounding introduced by [16, 19, 21].

Define a bipartite graph $G(x^*)=(V,E)$ where the vertices of $G$ are $V=\mathcal{J}{\cup}\mathcal{P}$ and $e=(i,j){\in}E$ if $x_{ij}^*{>}0$. The weight on edge $(i,j)$ is $x_{ij}^{*}w_j$. The rounding iteratively modifies $x_{ij}^*$, such that at the end $x_{ij}^*$ becomes integral. There are mainly two steps as following:

\textit{\romannumeral1. Break cycle:}

1.While$(G(x^*)$ has cycle $C=(e_1,e_2,...,e_{2l-1},e_{2l}))$

2.Set $C_1=(e_1,e_3,...,e_{2l-1})$ and $C_2=(e_2,e_4,...,e_{2l})$.

\quad Find minimal weight edge of $C$, denoted as $e_{min}^C$ and its weight ${\epsilon}=min_{e{\in}C_1||e{\in}C_2}e$;

3.If $e_{min}^C{\in}C_1$ then every edge in $C_1$ subtract ${\epsilon}$ and every edge in $C_2$ add ${\epsilon}$;

4.Else every edge in $C_1$ add ${\epsilon}$ and every edge in $C_2$ subtract ${\epsilon}$;

5.Remove the edges with weight $0$ from $G$.

\textit{\romannumeral2. Rounding fractional tasks:}

1.In the first rounding phase consider each integral assignment if $x_{ij}^*=1$, set $\bar{x}_{ij}=1$ and discard the corresponding edge from the graph. Denote again by $G$ the resulting graph;

2.While$(G(x^*)$ has connected component $C)$

3.Choose one task node from $C$ as root to construct a tree $Tr$, match each task node with any one of its children. The resulting matching covers all task nodes;

4.Match each task to one of its children node (a processor) such that $P_i=argmin_{P_i{\in}\mathcal{P}}\Sigma_{\bar{x}_{ij}=1}\bar{x}_{ij}w_j$, set $\bar{x}_{ij}=1$, and $\bar{x}_{ij}=0$ for other children node respectively.
\begin{lemma}
Relaxation-Dependent rounding finds an $2^{\alpha}$-approximation to the optimal schedule for the SEMRPP problem in the continuous model in polynomial time.
\end{lemma}
\begin{proof}
This can be concluded using the results of [19], we omit here.
\end{proof}
Next we improve this result by analyzing carefully for the SEMRPP problem by generalizing the result of \textit{Lemma $5$}.
\begin{theorem}
$(\romannumeral1)$ Relaxation-Dependent rounding finds an $2^{\alpha-1}(2-\frac{1}{p^{\alpha}})$-approximation to the optimal schedule for the SEMRPP problem in the continuous model in polynomial time, where $p=max_{\mathcal{M}_j}|\mathcal{M}_j|{\leq}m$. $(\romannumeral2)$ For any processor $P_i$, $\Sigma_{\mathcal{J}}\bar{x}_{ij}w_j<\Sigma_{\mathcal{J}}x_{ij}^*w_j+max_{\mathcal{J}:x_{ij}^*{\in}(0,1)}w_j$, $x_{ij}^*$ is the fractional task assignment at the beginning of the second phase. (i.e., extra maximal execution cycles linear constraints are violated only by $max_{\mathcal{J}:x_{ij}^*{\in}(0,1)}w_j$)
\end{theorem}
\begin{proof}
$(\romannumeral1)$ Denote the optimal solution for the SEMRPP problem as $OPT$, $H^*$ as the fractional schedule obtained after breaking all cycles and $\bar{H}$ as the schedule returned by the algorithm. Moreover, denote by $H_1$ the schedule consisting of the tasks assigned in the first step, i.e., $x_{ij}^*=1$ right after breaking the cycles and by $H_2$ the schedule consisting of the tasks assigned in the second rounding step, i.e., set $\bar{x}_{ij}=1$ by the matching process. We have $\|H_1\|_{\alpha}{\leq}\|H^*\|_{\alpha}{\leq}\|OPT\|_{\alpha}$ \footnote{In $H_1$ schedule, when the loads of $m$ processors is $\{l_1^{h1},l_2^{h1},...,l_m^{h1}\}$, $\|H_1\|_{\alpha}$ means $((l_1^{h1})^{\alpha}+(l_2^{h1})^{\alpha}+...+(l_m^{h1})^{\alpha})^{\frac{1}{\alpha}}$}, where the first inequality follows from the fact that $H_1$ is a sub-schedule of $H^*$ and the second inequality results from $H^*$ being a fractional optimal schedule compared with $OPT$ which is an integral schedule. We consider $\|H_1\|_{\alpha}{\leq}\|H^*\|_{\alpha}$ carefully. If $\|H_1\|_{\alpha}=\|H^*\|_{\alpha}$, that is all tasks have been assigned in the first step and the second rounding step is not necessary, then we have
$\|H_1\|_{\alpha}=\|H^*\|_{\alpha}=\|OPT\|_{\alpha}$. Such that the approximation is 1. Next we consider $\|H_1\|_{\alpha}<\|H^*\|_{\alpha}$, so there are some tasks assigned in the second rounding step, w.l.o.g., denote as $\mathcal{J}_1=\{J_1,...,J_k\}$. We assume the fraction of task $J_j$ assigned on processor $P_i$ is $f_{ij}$ and the largest eligible processor set size $p=max_{\mathcal{M}_j}|\mathcal{M}_j|{\leq}m$. Then we have
\begin{equation}
\begin{split}
(\|H^*\|_{\alpha})^{\alpha}
&=\sum_{i=1}^m(\Sigma_{J_j:x_{ij}^*=1}w_j+\Sigma_{J_j{\in}\mathcal{J}_1}f_{ij})^{\alpha}\\
&{\geq}\sum_{i=1}^m(\Sigma_{J_j:x_{ij}^*=1}w_j)^{\alpha}+\sum_{i=1}^m(\Sigma_{J_j{\in}\mathcal{J}_1}f_{ij})^{\alpha}\\
&=(\|H_1\|_{\alpha})^{\alpha}+\sum_{i=1}^m(\Sigma_{J_j{\in}\mathcal{J}_1}f_{ij})^{\alpha}\\
&{\geq}(\|H_1\|_{\alpha})^{\alpha}+\sum_{i=1}^{m}\sum_{j=1}^{k}(f_{ij})^{\alpha}\\
&=(\|H_1\|_{\alpha})^{\alpha}+\sum_{j=1}^{k}\sum_{i=1}^{m}(f_{ij})^{\alpha}\\
&{\geq}(\|H_1\|_{\alpha})^{\alpha}+\sum_{j=1}^k(\frac{\sum_{i=1}^{m}f_{ij}}{p})^{\alpha}\\
&=(\|H_1\|_{\alpha})^{\alpha}+\frac{1}{p^{\alpha}}\sum_{j=1}^{k}(w_j)^{\alpha}
\end{split}
\end{equation}
From the fact that $H_2$ schedules only one task per processor, thus optimal integral assignment for the subset of tasks it assigns and certainly has cost smaller than any integral assignment for the whole set of tasks. In a similar way we have
\begin{equation}
(\|H_2\|_{\alpha})^{\alpha}=\sum_{j=1}^k(w_j)^{\alpha}{\leq}(\|OPT\|_{\alpha})^{\alpha}
\end{equation}
So the inequality $(6)$ can be reduced to
\begin{equation}
(\|H^*\|_{\alpha})^{\alpha}{\geq}(\|H_1\|_{\alpha})^{\alpha}+\frac{1}{p^{\alpha}}(\|H_2\|_{\alpha})^{\alpha}
\end{equation}
then
\begin{equation*}
\begin{split}
(\|\bar{H}\|_{\alpha})^{\alpha}&=(\|H_1+H_2\|_{\alpha})^{\alpha}
{\leq}(\|H_1\|_{\alpha}+\|H_2\|_{\alpha})^{\alpha}\\
&=2^{\alpha}(\frac{\|H_1\|_{\alpha}+\|H_2\|_{\alpha}}{2})^{\alpha}\\
&{\leq}2^{\alpha}(\frac{1}{2}(\|H_1\|_{\alpha})^{\alpha}+\frac{1}{2}(\|H_2\|_{\alpha})^{\alpha})\\
&{\leq}2^{\alpha-1}((\|H^*\|_{\alpha})^{\alpha}-\frac{1}{p^{\alpha}}(\|H_2\|_{\alpha})^{\alpha}+(\|H_2\|_{\alpha})^{\alpha})\\
&{\leq}2^{\alpha-1}(2-\frac{1}{p^{\alpha}})(\|OPT\|_{\alpha})^{\alpha}
\end{split}
\end{equation*}
So
\begin{equation*}
\frac{(\|\bar{H}\|_{\alpha})^{\alpha}}{(\|OPT\|_{\alpha})^{\alpha}}{\leq}2^{\alpha-1}(2-\frac{1}{p^{\alpha}})
\end{equation*}
Which concludes the proof that the schedule $\bar{H}$  guarantees a $2^{\alpha-1}(2-\frac{1}{p^{\alpha}})$-approximation to optimal solution for the SEMRPP problem and can be found in polynomial time.

$(\romannumeral2)$ Seen from above, we also have
\begin{equation*}
\Sigma_{J_j{\in}\mathcal{J}}\bar{x}_{ij}w_j
<\Sigma_{J_j{\in}\mathcal{J}}x_{ij}^*w_j+max_{J_j{\in}\mathcal{J}:x_{ij}^*{\in}(0,1)}w_j,{\forall}P_i
\end{equation*}
Where the inequality results from the fact that the load of processor $P_i$ in $\bar{H}$ schedule is the load of $H^*$ plus the weight of task matched to it. Because we match each task to one of its child node, i.e., the execution cycle of the adding task $\bar{w}_j<max_{J_j{\in}\mathcal{J}:x_{ij}^*{\in}(0,1)}w_j$.
\end{proof}
Now we discuss the $s_{max}$. First we give Proposition $1$ to feasible and violation relationship.
\begin{proposition}
If $\mathbf{(P_1)}$ has feasible solution for the SEMRPP problem in the continuous model, we may hardly to solve $\mathbf{(P_1)}$ without violating the constraint of the limitation of the maximal execution cycles of processors.
\end{proposition}
Obviously, if $\mathbf{(P_1)}$ has a unique feasible solution, i.e., the maximal execution cycles of processors is set to the $OPT$ solution value. Then if we can always solve $\mathbf{(P_1)}$ without violating the constraint, this means we can easily devise an exact algorithm for $\mathbf{(P_1)}$. But we have proof that $\mathbf{(P_1)}$ is NP-Complete in the strong sense. Next, we give a guarantee speed which can be regarded as fast enough on the restricted parallel processors scheduling in the dependent rounding.
\begin{lemma}
Dependent rounding can get the approximation solution without violating the maximal execution cycles of processors constraint when\\ $s_{max}C{\geq}max_{P_i{\in}\mathcal{P}}L_i+max_{J_j{\in}\mathcal{J}}w_j$, where $L_i=\Sigma_{J_j{\in}\mathcal{J}_i}\frac{1}{|\mathcal{M}_j|}w_j$, $\mathcal{J}_i$ is the set of tasks that can be assigned to processor $P_i$.
\end{lemma}
\begin{proof}
First we denote a vector $\vec{H}=\{H_1,H_2,...,H_m\}$ in non-ascending sorted order as the execution cycles of $m$ processors at the beginning of the second step. We also denote a vector $\vec{L}=\{L_1,L_2,...,L_m\}$ in non-ascending sorted order as the execution of $m$ processors that  $L_i=\Sigma_{J_j{\in}\mathcal{J}_i}\frac{1}{|\mathcal{M}_j|}w_j$. Now we need to prove $H_1{\leq}L_1$. Suppose we have $H_1>L_1$, w.l.o.g., assume that the processor $P_1$ has the execution cycles of $H_1$. We denote the set of tasks assigned on $P_1$ as $\mathcal{J}_1^H$. Let $\mathcal{M}_1^H$ be the set of processors to which a task, currently fractional or integral assigned on processor $P_1$, can be assigned, i.e., $\mathcal{M}_1^H=\bigcup_{J_j{\in}\mathcal{J}_1^H}\mathcal{M}_j$. Similarly we denote the set of tasks can process on $\mathcal{M}_1^H$ as $\mathcal{J}^H$ and the set of processors $\mathcal{M}^H$ for every task in $P_i{\in}\mathcal{M}_1^H$ can be assigned, We have $\mathcal{M}^H=\bigcup_{J_j{\in}\mathcal{J}^H}\mathcal{M}_j$. W.l.o.g, we denote $\mathcal{M}^H$ as a set $\{h_1,h_2,...,h_k\}(1{\leq}k{\leq}m)$ and also denote a set $\{l_1,l_2,...,l_k\}(1{\leq}k{\leq}m)$ as its respective processors set in $\vec{L}$. According to the convexity of the objective, we get $H_{h_1}=H_{h_2}=...=H_{h_k}$. By our assumption, $H_{h_p}>L_{l_q}$,$\forall{p},\forall{q}$. Then
\begin{equation}
\Sigma_pH_{h_p}>\Sigma_qL_{l_q}
\end{equation}
Note that each integral task (at the beginning of the second step) in the left part of inequality $(9)$ can also have its respective integral task in the right part, but the right part may have some fractional task. So $\Sigma_qL_{l_q}-\Sigma_pH_{h_p}{\geq}0$, i.e., $\Sigma_pH_{h_p}{\leq}\Sigma_qL_{l_q}$, a contradiction to inequality $(9)$. The assumption is wrong, we have $H_1{\leq}L_1$. By Theorem $2$ the maximal execution cycles of dependent rounding $\bar{H}_{max}$, we have
\begin{equation*}
\begin{split}
\bar{H}_{max}&<H_1+max_{J_j{\in}\mathcal{J}:x_{ij}^*{\in}(0,1)}w_j\\
&{\leq}L_1+max_{J_j{\in}\mathcal{J}:x_{ij}^*{\in}(0,1)}w_j\\
&{\leq}L_1+max_{J_j{\in}\mathcal{J}}w_j=max_iL_i+max_{J_j{\in}\mathcal{J}}w_j
\end{split}
\end{equation*}
Finish the proof.
\end{proof}
\section{Experimental Results}
In this section, we provide performance detail of experimental results. To demonstrate the effectiveness of our approaches, we compare 5 values of interest, the optimal fractional solution, the optimal integral solution, the fractional dependent rounding integral (FDR, in the rest of paper, it refers to the solution of our algorithm) solution, the least flexible task (LFJ) solution and the least flexible processor (LFM) solution. We use the CPLEX solver [22] to obtain the optimal integral solution by solving the relevant Integer Programming. For our approximation algorithm, we obtain the optimal fractional solution by CVX solver [23], and then apply the dependent rounding by our algorithm. The results of LFJ and LFM solutions are obtained by following LFJ and LFM algorithms.

LFJ ALGORITHM. The tasks first are sorted in non-decreasing order of the cardinality of the processing sets of them, i.e., by $\vert\mathcal{M}_j\vert$. All the tasks are then scheduled in this order by sequential list. Next the task is assigned to a processor $P_i$ which has the least load and is in the task's processing set ($P_i\in{\mathcal{M}_j}$). At the last the speed of a processor is set to a value that the processor finishes its load by the time constraint; LFM ALGORITHM. The processors first are sorted in non-decreasing order of the cardinality of the processing task sets of them. The processors are then scheduled in this order by sequential list. Next the processor chooses a task which can be assigned on it and has not been assigned to other processors. At the last the speed of a processor is set to a value that the processor finishes its load by the time constraint. Note that the main difference between LFJ and LFM algorithm is the tasks or the processors as the object to select the processors or the tasks, correspondingly.
\subsection{Simulation Setting}
To evaluate the performance of our algorithm, we create systems consisting of 10 to 50 processors and 50 to 300 tasks. Each task $J_j$ is characterized by two parameters: the mount of the execution cycles $w_j$ and eligibility processing set $\mathcal{M}_j$. $w_j$ is randomly generated in the range $[1,10000]$. We simulate two case for $\mathcal{M}_j$. One is randomly generated from the set $\mathcal{P}$ of processors, and the other is arranged to construct the inclusive processing set restrictions\footnote{Inclusive processing set means that the pair restricted processing sets $\mathcal{M}_j$ and $\mathcal{M}_k$ for any two different tasks, either $\mathcal{M}_j{\subseteq}\mathcal{M}_k$ or $\mathcal{M}_k{\subseteq}\mathcal{M}_j$} [9]. Without loss of generality, the power parameter $\mathbf{\alpha}$ is set as $2$. The maximal speed $s_{max}$ is set to large enough to obtain the feasible solution. We analyse the effect of three different cases: the tightness of time constraint $C$, the ratio $\eta$ of the number of tasks to the number of processors, and the two different eligibility processing sets. All the results are mean values of different runs on an Intel Core I5-2400 CPU with $3.10GHz{\times}4$.
\subsection{Simulation Results}
Figure 1(a) represents the energy consumption of a 10 processors and 27 tasks system when the time constraint is increased. The five curves correspond to 5 values that we mention for comparing at the beginning of this section. Figure 1(b) reports the relative energy consumption ratio of these 5 values when all of them are normalized by the optimal integral. We find some observations from this simulation: $1)$. As shown in the Figure 1(a), 1(b), the energy consumption and the time constraint are in inverse proportion, and each ratio is almost not influenced by different time constraints. These confirms the \textit{Lemma 1} and \textit{Corollary 1}, i.e., each processor executes all tasks that are assigned on it at a uniform speed. So when the time constraint $C$ grows to $k{\times}C$, each processor can lower its speed to $\frac{s}{k}$ to finish the tasks. For $\mathbf{\alpha}=2$, the energy consumption is equal to $\frac{1}{k}$ (=$\frac{k{\times}C{\times}(\frac{s}{k})^2}{Cs^2}=\frac{1}{k}$) proportion of the energy consumption when the time constraint does not grow. Thus each kind energy consumption is influenced by the same proportion to the time constraint variation, when normalized by the optimal integral, the time constraint can be removed. This concludes the Figure 1(b). $2)$. The optimal fractional values are little different from the integral optimal. The Gap is at most $5\%$ in the experiment. This difference can also be observed between the integral optimal and the fractional dependent rounding integral solution, actually it is also within $5\%$ in the experiment. This suggests that the FDR performs much better than the approximation ratio we analysed in Theorem 2. $3)$. The figure confirms the superiority of the fractional dependent rounding integral solution, as it can reach $10\%$ better than the LFJ and LFM solution. After checking the maximum processor load, we find the result of the fractional dependent rounding is close to the integral optimal. This suggests the fractional dependent rounding integral solution can more efficiently balance the load between each eligibility processing set.

Figure 2(a) depicts the normalized energy consumption ratios for different solutions on varying ratios $\eta$ of the number of tasks to the number of processors. When the ratio $\eta$ is small, the difference between the normalized ratios is much larger. This can be explained by the fact that only one task be improperly assigned, the energy consumption would be excessively oscillated if $\eta$ is small. As the $\eta$ increasing, the shake will reduce because an improper task assignment will not influence so much. Figure 2(b) illustrates the normalized energy consumption ratios of a 14 processors and 35 tasks system for two eligibility processing sets. As shown in the figure, the different eligibility processing sets can influence the performance of the algorithms. The FDR and LFJ solution perform better in random processing set case. This can be explained by that in the LFJ and FDR (At the last stage when rounding fractional tasks to processors) solution the task chooses its processor, and the random restriction help the task do proper choice, but the difference is not so obvious. On the contrary, the LFM solution in which a processor chooses the tasks performs much better in inclusive processing set case. This can be explained by that the processor which has the less eligible tasks first select a task, if it does a improper choice, the subsequent processors will not influence much as they have more tasks to choose in inclusive processing set case. And it is interesting to observe that the algorithms perform much differently in random condition and regular condition.
\begin{figure}
\centering
\subfigure{
\begin{minipage}[t]{1\linewidth}
\includegraphics[width=2.4in]{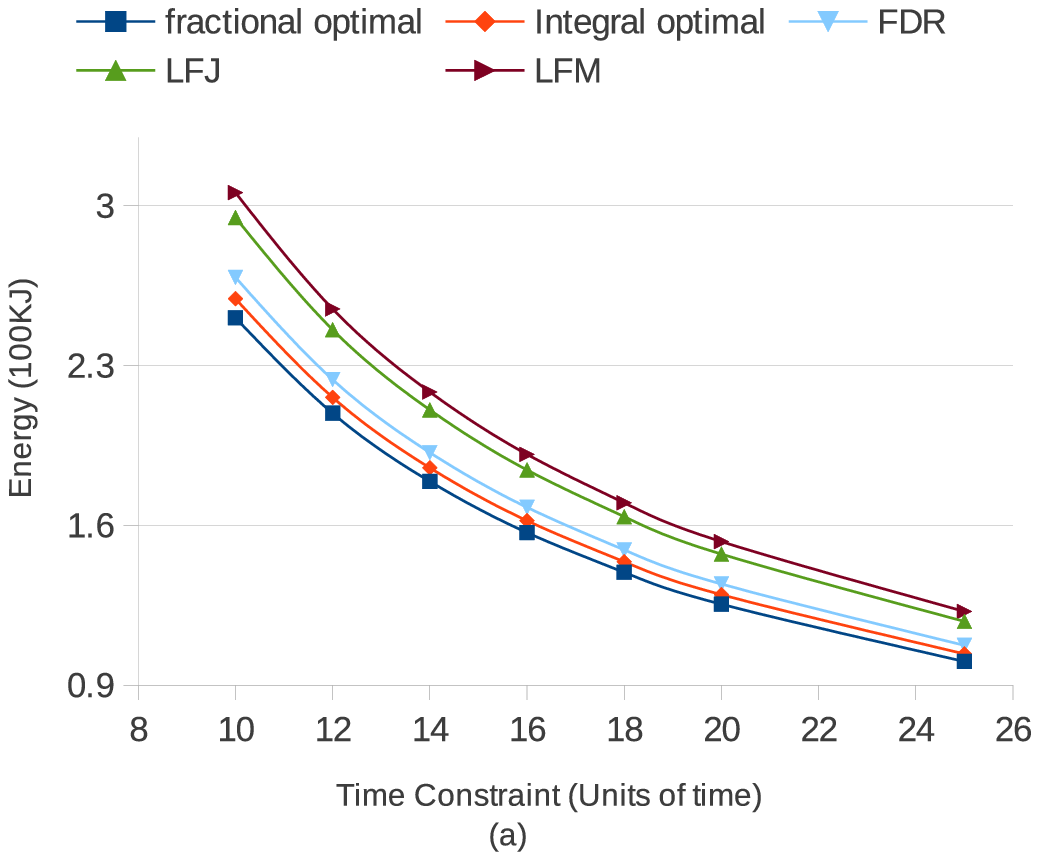}
\includegraphics[width=2.4in]{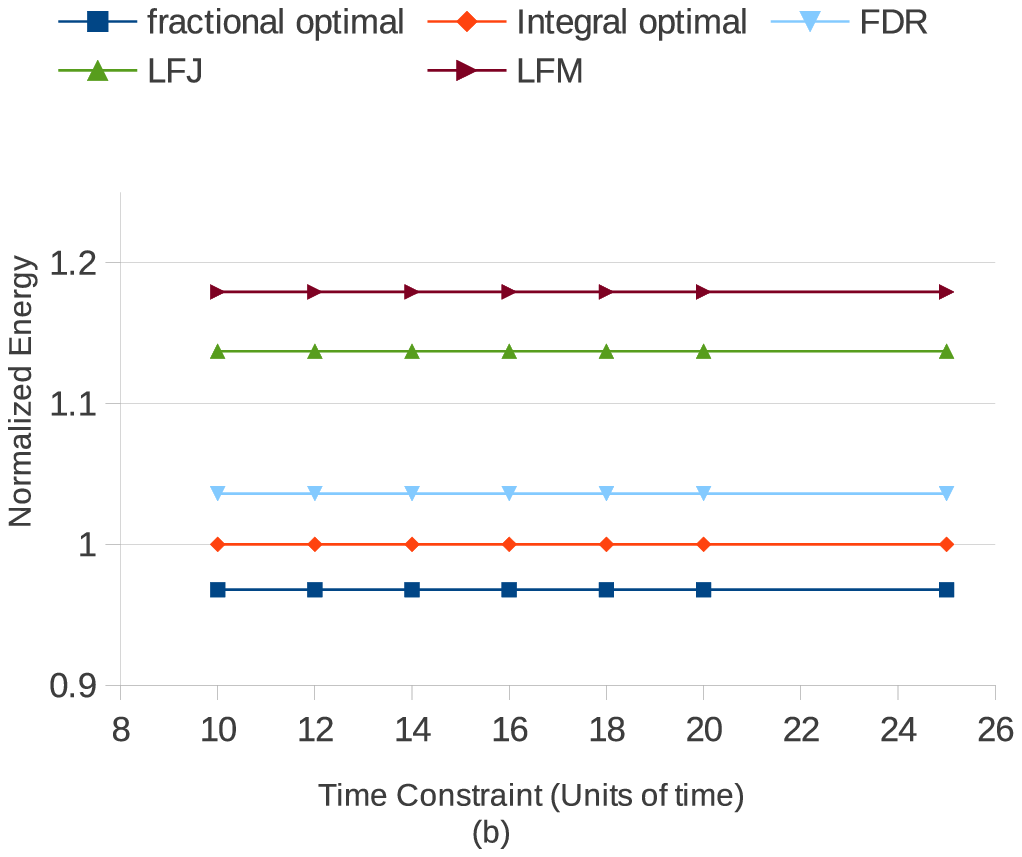}
\end{minipage}
}
\caption{(a) Energy consumption and (b) Normalized energy consumption ratio on time constraint.}\label{fig:graph}
\end{figure}
\begin{figure}
\centering
\subfigure{
\begin{minipage}[t]{1\linewidth}
\includegraphics[width=2.4in]{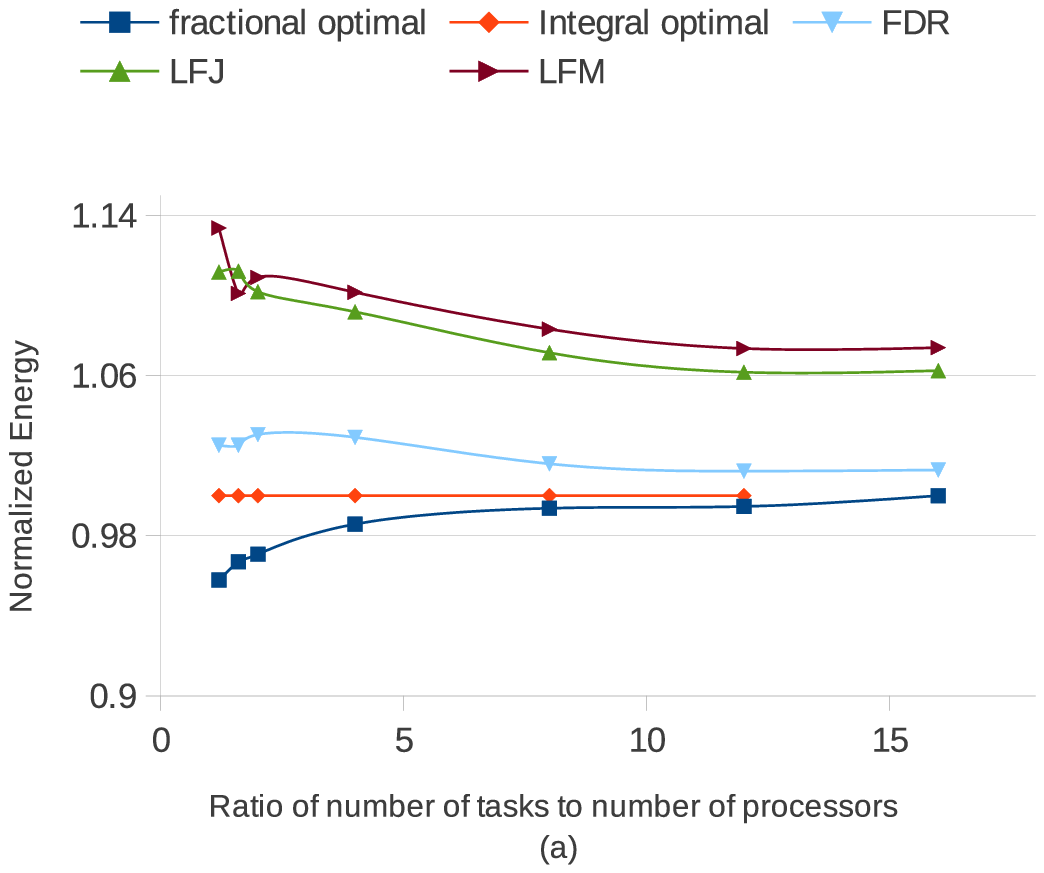}
\includegraphics[width=2.4in]{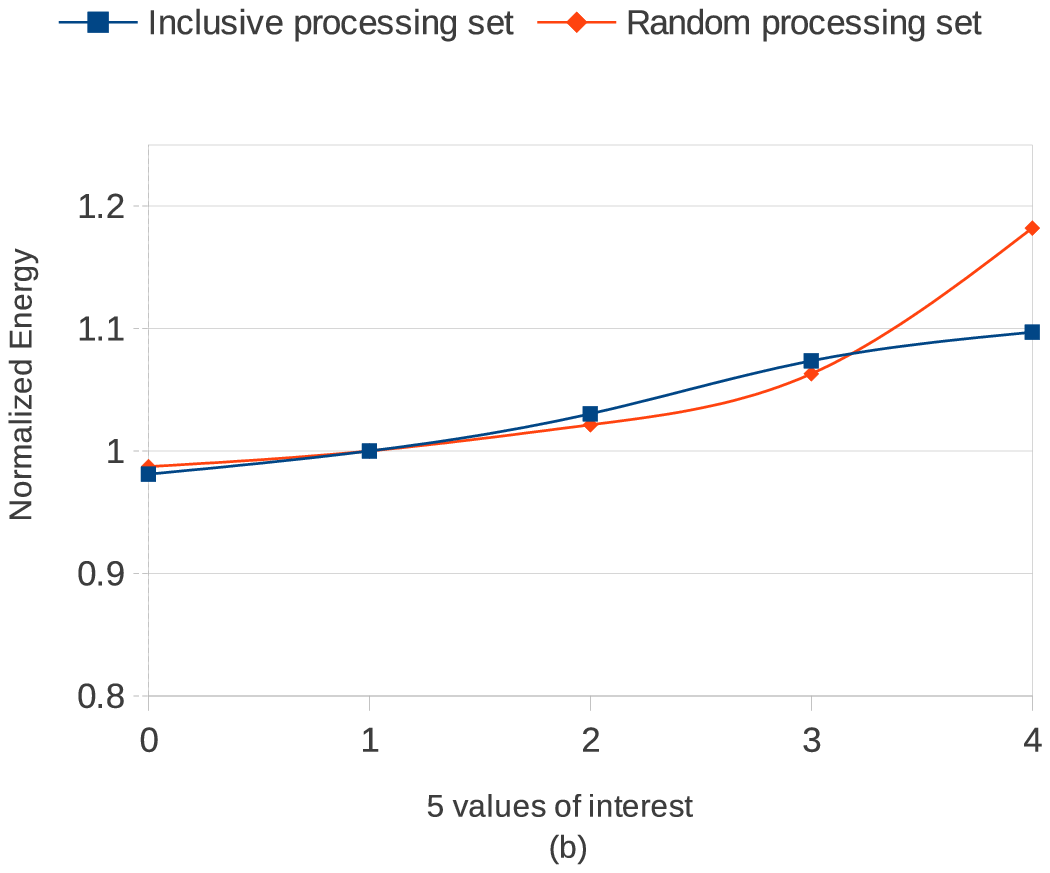}
\end{minipage}
}
\caption{(a) Normalized energy consumption ratio on varying ratios $\eta$ (The optimal integral value misses at the last point for it can not be obtained. The other values are normalized by the optimal fractional value.) and (b) Normalized energy consumption ratio on two eligibility processing sets (0-4 represent each value, respectively).}\label{fig:graph}
\end{figure}

The average running time for the optimal fractional solution solved by CVX, the fractional dependent rounding integral solution solved by CVX and rounding, the LFJ solution solved by LFJ algorithm and the LFM solution solved by LFM algorithm are fast (In our experiment it took at most several minutes) to all the instances presented so far. But the optimal integral solution solved by CPLEX takes more than one day in large systems. For larger systems, the optimal integral solution has trouble in both memory and running time. Note that during all the experiments, the FDR solution is efficient than LFJ and LFM solution. This suggests that our solution could assign tasks more properly in every instance, and solve the SEMRPP problem efficiently due to high quality and low computational time.

We emphasize that, as per the latest reports [24, 25], every year the energy costs are on the order of billions of dollars. Given this, a reduction by even a few percent in energy cost can result in savings of billions of dollars.
\section{Conclusion}
In this paper we explore algorithmic instruments leading to reduce energy consumption on restricted parallel processors. We aim at minimizing the sum of energy consumption while the speed scaling method is used to reduce energy consumption under the execution time constraint $(C_{max})$. We first assess the complexity of scheduling problem under speed and restricted parallel processors settings. We present a polynomial-time approximation algorithm with a $2^{\alpha-1}(2-\frac{1}{p^{\alpha}})$-approximation $(p=max_{\mathcal{M}_j}|\mathcal{M}_j|{\leq}m)$ factor for the general case that the tasks have arbitrary size of execution cycles. Specially, when the tasks have a uniform size, we propose an optimal scheduling algorithm with time complexity $O(mn^3logn)$. We evaluate the performance of our algorithm by a set of simulated experiments. It turns out that our solution is very close to the optimal solution. This confirms our algorithm could provide efficient scheduling for the SEMRPP problem.

\end{document}